\documentclass[sigconf]{acmart}
\renewcommand\footnotetextcopyrightpermission[1]{} % removes footnote with conference information in the first column
\AtBeginDocument{%
  }

\usepackage{graphicx}
\usepackage{subcaption}
\usepackage{multirow}
\usepackage[linesnumbered,ruled]{algorithm2e}
\usepackage{amsmath}
\allowdisplaybreaks
\begin{document}

\setlength{\abovedisplayskip}{1.5pt}
\setlength{\belowdisplayskip}{1.5pt}
\setlength{\abovedisplayshortskip}{1.5pt}
\setlength{\belowdisplayshortskip}{1.5pt}

\setlength{\topsep}{0pt}
\setlength{\parskip}{0pt}

\title{Addressing Personalized Bias for Unbiased Learning to Rank}

\author{Zechun Niu}
\affiliation{%
  \department{Gaoling School of Artificial Intelligence}
  \institution{Renmin University of China}
  \city{Beijing}
  \country{China}
}
\email{niuzechun@ruc.edu.cn}
\orcid{0009-0009-7954-3713}

\author{Lang Mei}
\affiliation{%
  \department{Gaoling School of Artificial Intelligence}
  \institution{Renmin University of China}
  \city{Beijing}
  \country{China}
}
\email{meilang2013@ruc.edu.cn}
\orcid{0000-0002-7960-3036}

\author{Liu Yang}
\affiliation{%
  \department{Gaoling School of Artificial Intelligence}
  \institution{Renmin University of China}
  \city{Beijing}
  \country{China}
}
\email{yangliusuper6@gmail.com}
\orcid{0000-0002-0154-9957}

\author{Ziyuan Zhao}
\affiliation{%
  \institution{Search Algorithm Group, WeChat, Tencent}
  \city{Guangzhou}
  \state{Guangdong}
  \country{China}
}
\email{joshuazhao@tencent.com}
\orcid{0009-0000-4742-4116}

\author{Qiang Yan}
\affiliation{%
  \institution{Search Algorithm Group, WeChat, Tencent}
  \city{Guangzhou}
  \state{Guangdong}
  \country{China}
}
\email{rolanyan@tencent.com}
\orcid{0000-0002-7328-2278}

\author{Jiaxin Mao}
\authornote{Corresponding author.}
\affiliation{%
  \department{Gaoling School of Artificial Intelligence}
  \institution{Renmin University of China}
  \city{Beijing}
  \country{China}
}
\email{maojiaxin@gmail.com}
\orcid{0000-0002-9257-5498}

\author{Ji-Rong Wen}
\affiliation{%
  \department{Gaoling School of Artificial Intelligence}
  \institution{Renmin University of China}
  \city{Beijing}
  \country{China}
}
\email{jrwen@ruc.edu.cn}
\orcid{0000-0002-9777-9676}

\renewcommand{\shortauthors}{Zechun Niu et al.}

\begin{abstract}
  Unbiased learning to rank (ULTR), which aims to learn unbiased ranking models from biased user behavior logs, plays an important role in Web search. Previous research on ULTR has studied a variety of biases in users' clicks, such as position bias, presentation bias, and outlier bias. However, existing work often assumes that the behavior logs are collected from an ``average'' user, neglecting the differences between different users in their search and browsing behaviors. In this paper, we introduce personalized factors into the ULTR framework, which we term the user-aware ULTR problem. Through a formal causal analysis of this problem, we demonstrate that existing user-oblivious methods are biased when different users have different preferences over queries and personalized propensities of examining documents. To address such a personalized bias, we propose a novel user-aware inverse-propensity-score estimator for learning-to-rank objectives. Specifically, our approach models the distribution of user browsing behaviors for each query and aggregates user-weighted examination probabilities to determine propensities. We theoretically prove that the user-aware estimator is unbiased under some mild assumptions and shows lower variance compared to the straightforward way of calculating a user-dependent propensity for each impression.\footnote{In this paper, an ``impression'' refers to an occurrence of a document in a session.} Finally, we empirically verify the effectiveness of our user-aware estimator by conducting extensive experiments on two semi-synthetic datasets and a real-world dataset.
\end{abstract}

\begin{CCSXML}
<ccs2012>
<concept>
<concept_id>10002951.10003317.10003338.10003343</concept_id>
<concept_desc>Information systems~Learning to rank</concept_desc>
<concept_significance>500</concept_significance>
</concept>
</ccs2012>
\end{CCSXML}

\ccsdesc[500]{Information systems~Learning to rank}

\keywords{Unbiased Learning to Rank; Personalized Bias; Inverse Propensity Score}

\maketitle

\section{Introduction}
\begin{figure}[t]
  \centering
  \vspace{-3mm}
  \includegraphics[width=0.90\linewidth,trim=200 200 200 150, clip]{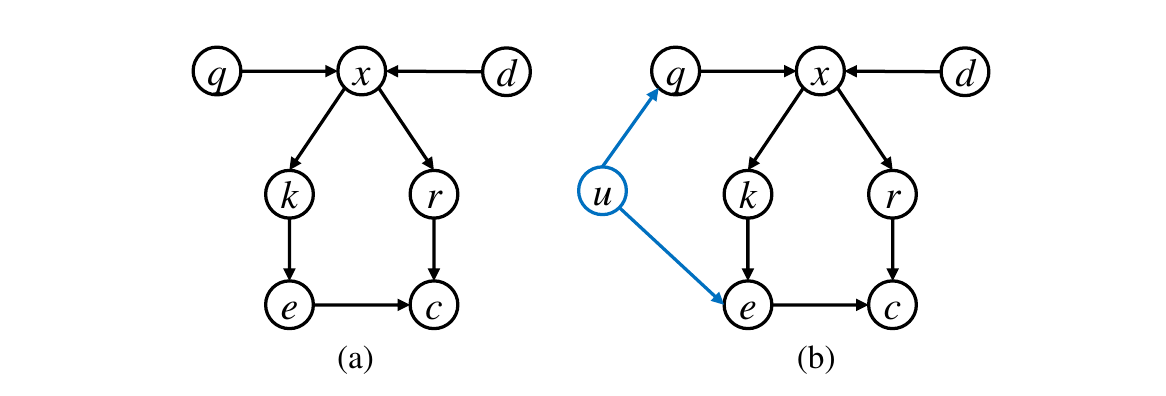}
  \vspace{-6mm}
  \caption{(a) The causal graph of existing ULTR methods \cite{wang2016ips, joachims2017ips} for generating observational click data. (b) The causal graph of the user-aware click generation process, where the existence of $u$ creates an additional backdoor path from $e$ to $c$: $e \leftarrow u \rightarrow q \rightarrow x \rightarrow r \rightarrow c$. The meanings of the variables are as follows. $q$: query, $x$: query-document features, $d$: document, $k$: position in the ranked list, $r$: relevance judgment, $e$: examination, $c$: click, $u$: user.}
  \label{fig:causal_graphs}
  \vspace{-6mm}
  \Description{This figure shows the causal graphs.}
\end{figure}

Hundreds of millions of users interact with Web search engines every day, generating massive amounts of user behavior data such as clicks and dwell time. User behavior data are valuable because they can implicitly reflect the relevance of the documents returned by the search engines. By leveraging user behavior logs, we can train a ranking model without expensive human relevance annotations. However, the user behavior data are intrinsically biased \cite{joachims2005accurately, joachims2007evaluating, o2006modeling, yue2010beyond}. For example, documents that are not clicked are not necessarily irrelevant; they may just not be examined by users. Therefore, directly using user click data as a supervision signal for training ranking models can lead to suboptimal performance. To solve this problem, many unbiased learning to rank (ULTR) models \cite{wang2016ips, joachims2017ips, ai2018dla, wang2018regression-EM, vardasbi2020affine, niu2025distributionally} have been proposed to learn an unbiased ranking model from biased user click logs. 

Typically, ULTR models are based on an examination hypothesis (EH) \cite{richardson2007ExminationHypothesis} that a user would click a document if and only if he/she both examines it and perceives it as relevant to the query, as shown in Figure~\ref{fig:causal_graphs}(a) where $c$ is affected by $e$ and $r$. In other words, the behavior of clicking can be broken down into two parts: examination and relevance judgment. The key idea of ULTR is to remove the impact of unevenly distributed examination on click-or-not signals to mitigate biases. Specifically, the ULTR models first calculate the probabilities of documents being examined in the search sessions based on certain user behavior assumptions.\footnote{In this paper, a ``session'' refers to the set of events between a user's issuing a query and abandoning the search result page. Note that the session here applies only to one query, not a set of queries related to the same information need.} Then, they utilize the inverse propensity score (IPS) \cite{ipw} method to reweigh the clicked documents with the reciprocal of their examination probabilities (called ``propensities'') to recover the true relevance. It can be proved that the IPS method can obtain an unbiased estimate of any ranking objective when the user behavior assumptions are correct and the propensity estimation is accurate
 \cite{joachims2017ips}.

Previous work on ULTR has studied a variety of biases involving different factors that affect users' examination probabilities, such as position \cite{joachims2017ips}, presentation features \cite{mao2022wholepageULTR}, and outlier \cite{sarvi2023outlierbias}. However, existing ULTR methods often implicitly assume that each interaction record comes from an ``average'' user and ignore the differences in users' search and browsing behaviors. Specifically, they do not explicitly model the idiosyncrasies of users and utilize the same user behavior model to compute the propensities for all sessions from different users. On the contrary, in real-world Web search scenarios, different users usually issue different queries because they have personalized interests and information needs \cite{sharifpour2023analysis-of-query-logs} or use different devices (e.g., computers or mobile phones) \cite{montanez2014cross-device_search}. Besides, Zhang et al. \cite{zhang2022globalorlocalCM} showed that different users also browse the search engine result pages in personalized ways. For example, some users are severely affected by the position and only focus on the top documents, while some users are more patient and will examine documents further down the list. To the best of our knowledge, few existing studies have investigated whether and how the differences between different users would affect the unbiasedness and variance of existing ULTR methods.

To fill the gap, in this paper, we aim to study the ULTR problem that involves users' personalized behaviors, i.e., the user-aware ULTR problem. To clearly formulate this problem, we first draw a new causal graph for the generation process of user behavior data as shown in Figure~\ref{fig:causal_graphs}(b), where the user variable (marked in blue) is included and affects both the query and examination variables. Through causal analysis, we show that the combination of users' personalized behaviors of issuing queries and examining the search results creates an additional backdoor path from examination to click, leading to a personalized bias that is not yet addressed by existing user-oblivious ULTR methods. Intuitively, we can mitigate personalized bias via a straightforward approach that calculates a user-dependent propensity for each impression and utilizes the IPS method to estimate learning-to-rank objectives. However, this straightforward estimator suffers from high variance when some users' examination probabilities are relatively small.

To solve the user-aware ULTR problem, we propose a novel user-aware inverse-propensity-score estimator for learning-to-rank objectives that balances the examination distributions over documents under each query rather than under each session. Specifically, the user-aware estimator models the user distributions under different queries and the examination probabilities of different users, and for each query aggregates the user-weighted examination probabilities as the propensities. We theoretically prove that the user-aware estimator is unbiased if each relevant document has a non-zero probability of being examined. Besides, we theoretically show that the user-aware estimator has a lower variance than the straightforward estimator mentioned above. 

We conduct extensive experiments on two semi-synthetic datasets and a real-world dataset to verify the effectiveness of our user-aware estimator. The results show that the user-aware estimator is significantly superior to the existing ULTR methods when facing personalized bias, and outperforms the straightforward estimator, especially when there are fewer behavior logs. 

The contributions of this paper are summarized as follows:
\begin{itemize}
\item To the best of our knowledge, we are the first to formulate the problem of user-aware unbiased learning to rank from a causal perspective and identify an unresolved personalized bias caused by the combination of users' personalized behaviors of issuing queries and examining documents.
\item To mitigate the personalized bias, we propose a novel user-aware inverse-propensity-score estimator for the objective functions of learning-to-rank models that balances the examination distributions over documents under each query. 
\item We theoretically prove that the user-aware estimator is unbiased under some mild assumptions and shows lower variance compared to the straightforward estimator.
\item We empirically verify that our user-aware estimator is superior to the existing ULTR methods and the straightforward estimator when facing personalized bias.
\end{itemize}

\vspace{-4mm}
\section{Unbiased Learning to Rank}

\subsection{Traditional Learning to Rank}
Traditional learning to rank (LTR) \cite{liu2009learningtorank} uses human relevance annotations to train ranking models. Let $q$ be a query from a set of training queries $Q$, let $S$ be a ranking model, and let $l(S|q)$ be a defined loss function over $S$ and $q$. The goal of LTR is to minimize the empirical risk \cite{Vapnik1995empiricalrisk} of $S$ over $Q$:
% empirical loss
\begin{align}
    \mathcal{L}(S) = \frac{1}{|Q|}\sum_{q \in Q}l(S|q)
\end{align}
The ranking loss $l(S|q)$ is usually related to a certain ranking metric (e.g., ARP, nDCG\cite{jarvelin2002cumulated}, and ERR \cite{chapelle2009expected}). Let $\pi_q$ be the ranked list produced by $S$ for query $q$, let $d$ be a document in $\pi_q$, and let $y$ be the relevance of $d$ to $q$. Then, $l(S|q)$ often takes the form of a sum over documents:
% unbiased estimator
\begin{equation}
\label{eq:unbiased_estimator}
    l(S|q) = \sum_{d\in \pi_q}\lambda(d|\pi_q) \cdot \sigma(y(d)),
\end{equation}
where $\lambda(d|\pi_q)$ and $\sigma(y(d))$ are usually chosen for a specific ranking metric. For example, if we use the Average Relevance Position (ARP) metric, then $\lambda(d|\pi_q)$ should be the position of the document in the ranked list, i.e., $k(d|\pi_q)$ and $\sigma(y(d))$ should be $y(d)$. 
% As for the Discounted Cumulative Gain (DCG) metric, $\lambda(d|\pi_q)$ should be $-1/{\rm log}_2(k(d|\pi_q)+1)$ and $\sigma(y(d))$ should be $2^{y(d) - 1}$.
\vspace{-3mm}

\subsection{Counterfactual Learning to Rank}
\label{sec:CLTR}
% Users' clicks are cheap to collect at scale and can reflect users' preferences for relevance.
Since getting human relevance annotations is expensive and time-consuming, learning a ranking model from logged user behavior data (e.g., clicks) becomes an alternative. Given a set of logged sessions $N$ for query $q$ collected from past interactions between a user and the online search system, a naive way to estimate the local ranking loss is to directly utilize clicks as relevance labels: 
\begin{align}
\label{eq:naive_estimator}
\begin{split}
\hat{l}_{naive}(S|q,N_q) =\frac{1}{|N_q|}\sum_{\vec{c} \in N_q} \hat{l}_{naive}(S, q|\vec{c}) \\
\hat{l}_{naive}(S| q, \vec{c}) = \sum_{d\in \pi_q}\lambda(d|\pi_q) \cdot c(d) ,
\end{split}
\end{align}
where $\vec{c}$ denotes a logged click list.

Unfortunately, users' clicks are not equivalent to true relevance because they are also affected by the exposure of the documents according to the widely adopted examination hypothesis (EH):
\begin{equation}
\label{eq:EH}
    c(d) = e(d) \cdot r(d),
\end{equation}
where $c$, $e$, and $r$ are binary variables denoting whether a document $d$ is clicked, examined, and judged as relevant, respectively. For simplicity, in the theoretical derivation of this paper, we suppose that $\sigma(y(d))$ is also binary and complies with users' relevance judgment $r$. Hence, the ideal loss in Eq.~(\ref{eq:unbiased_estimator}) can be rewritten as:
\begin{equation}
\label{eq:ideal_estimator}
    l_{ideal}(S|q) = \sum_{d\in \pi_q}\lambda(d|\pi_q) \cdot r(d)
\end{equation}
Note that this assumption does not lose generality because it is easy to extend to any case where the true relevance $y$ is mapped to the probability of being perceived as relevant $P(r=1)$ in a one-to-one manner. Because the examination distributions over the documents in the click logs are usually uneven, clicks are biased toward examinations. For example, users are more inclined to examine the top-ranked documents, which would lead to a position bias \cite{craswell2008experimental}. Therefore, the naive way of Eq.~(\ref{eq:naive_estimator}) would obtain biased estimates.
% and it is necessary to develop unbiased learning to rank methods to infer the true relevance from biased users' clicks.

To mitigate the position bias in users' clicks, Wang et al. \cite{wang2016ips} and Joachims et al. \cite{joachims2017ips} proposed the first counterfactual learning to rank (CLTR) method based on Inverse Propensity Score (IPS). They implicitly modeled the generation process of click data with the causal graph shown in Figure~\ref{fig:causal_graphs}(a).\footnote{The original papers \cite{wang2016ips, joachims2017ips} did not explicitly draw a causal graph and did not elaborate on the debiasing theory from the perspective of causal inference.} In the causal graph, $x$ represents the ranking features for the query-document pair, including the features about the query $q$, the document $d$, and the match between $q$ and $d$. The path $x \rightarrow k$ stands because the online production ranker takes the ranking features as input and then decides the positions of the candidate documents in the logged list $R_q$.\footnote{Following the settings of Joachims et al. \cite{joachims2017ips}, we assume that the online production ranker is deterministic and always returns a unique ranked list for each query.} Besides, they adopted the position-based click model (PBM) \cite{craswell2008experimental} and assumed that examination $e$ is only affected by position $k$. The mechanism that causes position bias is the existence of the backdoor path $e \leftarrow k \leftarrow x \rightarrow r \rightarrow c$. To reveal true relevance from clicks, they proposed the IPS-PBM \cite{wang2016ips, joachims2017ips} estimator to answer such a question: for a document in a specific search session, would it be clicked if it were examined? Given a logged click list $\vec{c}$ for query $q$, IPS-PBM estimates the local ranking loss as:
\begin{equation}
\label{eq:IPS_PBM_estimator}
\begin{split}
    % & \hat{l}_{IPS-PBM}(S, q|N_q) =\frac{1}{|N_q|}\sum_{n \in N_q} \hat{l}_{IPS-PBM}(S, q|\vec{c_n}) \\
    % & 
    \hat{l}_{IPS-PBM}(S|q, \vec{c}) = \sum_{d\in \pi_q} \frac{\lambda(d|\pi_q)c(d)}{P(e(d)=1|k(d|R_q))},
\end{split}
\end{equation}
where $k(d|R_q)$ denotes the position of the document in the logged ranked list $R_q$.\footnote{Note that $R_q$ is returned by the online production ranker and $\pi_q$ is returned by $S$.} For simplicity, we will abbreviate $k(d|R_q)$ as $k(d)$ in the following paper. By using the position-dependent examination probabilities $P(e(d)=1|k(d))$ as the propensities, IPS-PBM cuts off the path from $k$ to $e$, thereby blocking the backdoor path $e \leftarrow k \leftarrow x \rightarrow r \rightarrow c$. Joachims et al. \cite{joachims2017ips} proved that the IPS-PBM estimator is unbiased if the examination probability of each relevant document is greater than zero. 

Researchers have extended the above IPS-based CLTR method to address many other biases in users' click logs, such as presentation bias \cite{mao2022wholepageULTR} and outlier bias \cite{sarvi2023outlierbias}, by conditioning examination probabilities also on other factors (e.g., outliers and other SERP features). However, they do not consider the differences between users. That is, they all implicitly assume that the click logs were generated by an ``average'' user. Taking the IPS-PBM \cite{wang2016ips,joachims2017ips} method as an example, it first uses an online randomization experiment to estimate the overall examination probabilities for different positions. Then, it leverages these overall estimates as propensities for all the clicks that are actually from different users. Therefore, the existing user-oblivious ULTR methods lack a theoretical guarantee of unbiasedness when users have personalized behaviors.

\section{User-aware Unbiased Learning to Rank}
In this paper, we introduce users' personalized search and browsing behaviors into the ULTR framework, which we term the user-aware ULTR problem. 

\begin{figure}[t]
  \centering
  \vspace{-3mm}
  \includegraphics[width=\linewidth,trim=100 200 100 150, clip]{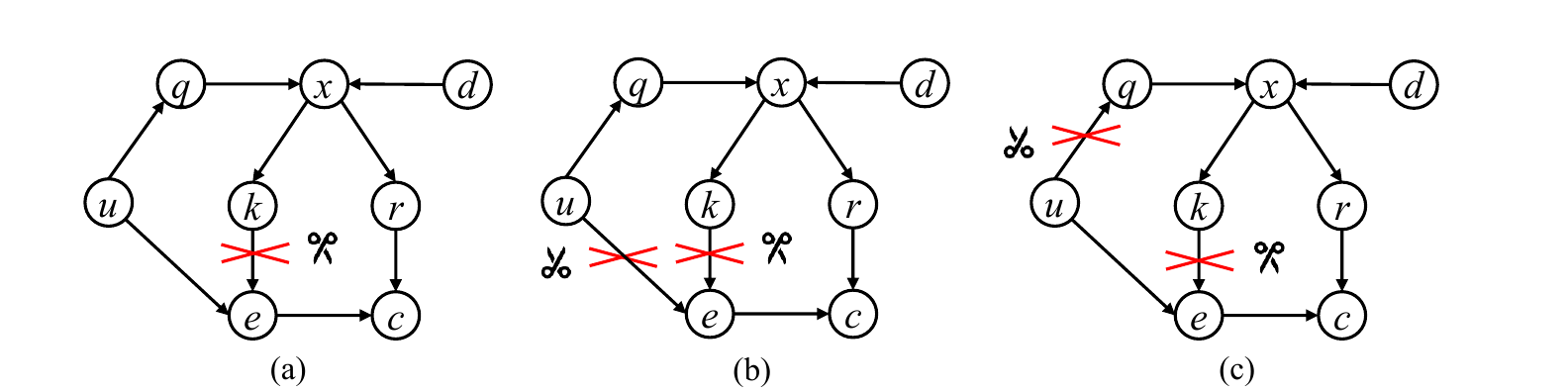}
  \vspace{-8mm}
  \caption{(a) (b) (c) Illustrations of the causal interventions simulated by IPS-PBM \cite{wang2016ips,joachims2017ips}, the straightforward estimator, and our user-aware estimator, respectively. The meanings of the variables are the same as those in Figure~\ref{fig:causal_graphs}.}
  % as follows. $q$: query, $x$: query-document features, $d$: document, $k$: position of the document in the ranked list, $r$: relevance judgment, $e$: examination, $c$: click, $u$: user.
  \label{fig:causal_graphs_ipsupbm}
  \vspace{-6.5mm}
  \Description{This figure shows the causal interventions.}
\end{figure}

\subsection{User-oblivious Estimators Suffer from Personalized Bias}
\label{sec:personalized-bias}
To clearly formulate the user-aware ULTR problem, we first propose an augmented causal graph to model the generation process of users' click logs, as shown in Figure~\ref{fig:causal_graphs}(b). Since previous research \cite{montanez2014cross-device_search, sharifpour2023analysis-of-query-logs, zhang2022globalorlocalCM} suggest that users have personalized behaviors of issuing queries and examining documents, there are paths from $u$ to $q$ and from $u$ to $e$ (marked in blue) in the causal graph. Besides, we assume that each user follows a PBM model when examining the documents, so examination $e$ is only affected by user $u$ and position $k$. Note that we use PBM only for simplicity and alignment with previous work \cite{wang2016ips, joachims2017ips}, and this assumption can be easily extended to other click models, such as DCM\cite{guo2009dcm} and UBM\cite{dupret2008ubm}, or even situations where different users follow different click models. In addition, we follow the conventions in Web search and ULTR to assume that users have the same relevance judgments for given query-document pairs, so $r$ is only affected by $x$. Although it is also natural to assume that different users have different relevance judgments in some scenarios, we leave it to future work. 

From the causal graph, we can see that the existence of $u$ creates an additional backdoor path from $e$ to $c$: $e \leftarrow u \rightarrow q \rightarrow x \rightarrow r \rightarrow c$. Since IPS-PBM uses the overall position-dependent examination probabilities as propensities, it can be rewritten into:
% IPS-PBM estimator with u
\begin{equation}
    \hat{l}_{IPS-PBM}(S|q,\vec{c}) = \sum_{d\in \pi_q} \frac{\lambda(d|\pi_q)c(d)}{\sum_{u \in U} P(e(d)=1|k(d),u)P(u)},
\end{equation}
where $U$ represents the whole user set. That is, IPS-PBM only balances the distribution of examination to position and cuts off the path $k \rightarrow e$ as shown in Figure~\ref{fig:causal_graphs_ipsupbm}a. Therefore, IPS-PBM can not block the additional backdoor path caused by $u$ and thus suffers from an unsolved bias in the user-aware ULTR problem, which we term \emph{personalized bias}. 

To better show that the user-oblivious IPS-PBM does suffer from personalized bias, we introduce a simple example. Suppose that there is a query $q'$ issued by only one user and the user's examination probabilities for positions differ from the average of all users. For instance, this user has a probability of $0.9$ to examine the first ranked document and $0.5$ for the second, while the average of all users is $0.9$ and $0.4$. In this case, when estimating the local loss for $q'$, IPS-PBM would utilize the wrong propensity for the second-ranked document and yield biased relevance estimates. 

\textbf{Remark}. Users' personalized search behaviors \cite{montanez2014cross-device_search, sharifpour2023analysis-of-query-logs} play an important role in the establishment of personalized bias. If users only have personalized browsing behaviors but no personalized search behaviors (i.e., there is no path from $u$ to $q$ in the causal graph), then the backdoor path $e \leftarrow u \rightarrow q \rightarrow x \rightarrow r \rightarrow c$ does not exist and IPS-PBM is unbiased.

% (We will discuss this from a computational perspective later in Section~\ref{sec:user-aware estimator}.)

\subsection{Straightforward Estimator}
A straightforward way to address the personalized bias is to condition each document's examination probability in each session not only on its position but also on the specific user. From a causal view, it answers the same counterfactual question as that of IPS-PBM mentioned in Section~\ref{sec:CLTR}: for a document in a specific search session, would it be clicked if it were examined (by the user corresponding to this session)? That is, we can first identify the user corresponding to each session and then utilize his/her examination probabilities as the propensities:
\begin{equation}
\label{eq:straightforward estimator}
\begin{split}
    % & \hat{l}_{straight}(S, q|N_q) =\frac{1}{|N_q|}\sum_{n \in N_q} \hat{l}_{straight}(S, q|\vec{c_n}) \\
    \hat{l}_{straight}(S|q, \vec{c}, u) = \sum_{d\in \pi_q} \frac{\lambda(d|\pi_q)c(d)}{P(e(d)=1|k(d),u)} \quad,
\end{split}
\end{equation}
where $u$ represents the user of the session. This straightforward estimator can be regarded as a simple extension of IPS-PBM that considers the specific user corresponding to each session. Following Joachims et al. \cite{joachims2017ips}, it is easy to prove that the straightforward estimator is unbiased if each relevant document has a positive probability of being examined in each session. However, it suffers from high variance. Given the logged session $N_q$, its variance is:
\begin{align}
\label{eq:variance_straight}
% \begin{split}
    \mathbb{V}& \left[ \hat{l}_{straight}(S|q,N_q) \right]  = \frac{1}{|N_q|^2} \sum_{(\vec{c}, u) \in N_q}\mathbb{V} \left[\hat{l}_{straight}(S|q, \vec{c}, u) \right] \nonumber \displaybreak[1]\\
    % &= \frac{1}{|N_q|^2} \sum_{(\vec{c}, u) \in N_q} \mathbb{V} \left[ \sum_{d\in \pi_q} \frac{\lambda(d|\pi_q)c(d)}{P(e(d)=1|k(d),u)} \right] \nonumber \displaybreak[1]\\
    &= \frac{1}{|N_q|^2} \sum_{(\vec{c}, u) \in N_q} \sum_{d\in \pi_q}  \mathbb{V} \left[  \frac{\lambda(d|\pi_q)c(d)}{P(e(d)=1|k(d),u)} \right] \nonumber \displaybreak[1]\\
    &= \frac{1}{|N_q|^2} \sum_{d\in \pi_q} \sum_{(\vec{c}, u) \in N_q} \frac{\lambda^{2}(d|\pi_q)}{(P(e(d)=1|k(d),u))^{2}}  \mathbb{V} \left[  c(d) \right] \nonumber \displaybreak[1]\\
    &= \frac{1}{|N_q|^2} \sum_{d\in \pi_q :r(d)=1} \sum_{(\vec{c}, u) \in N_q} \frac{\lambda^2(d|\pi_q) (1 - P(e(d)=1|k(d),u))}{P(e(d)=1|k(d),u)} \nonumber \displaybreak[1]
% \end{split}
\end{align}
The last step utilizes the examination hypothesis in Eq.~(\ref{eq:EH}). If some users' examination probabilities for some documents $P(e(d)=1|k(d),u)$ are small, then the variance of this estimator will be large, which may distort the training of the ranking model.

\vspace{-4mm}

\subsection{User-aware Inverse-propensity-score Estimator}
\label{sec:user-aware estimator}
Now we describe our novel user-aware inverse-propensity-score estimator that is both unbiased and has a lower variance. The key idea of reducing variance while remaining unbiased is to mitigate the personalized bias by balancing the examination distributions over documents under each query rather than under each session. 

From a causal view, we answer a different question from that mentioned in Section~\ref{sec:CLTR}: would a document be clicked if we ``force'' it to be examined (regardless of the specific user or session)? For a certain query-document pair $x_0$, we infer its relevance as follows:
\begin{align}
% \begin{split}
    &P(r=1|x=x_0) = P(c=1|do(e=1),x=x_0) \nonumber \displaybreak[1]\\
    &= P(c=1|e=1,x=x_0) = \frac{P(c=1,e=1,x=x_0)}{P(e=1,x=x_0)} \nonumber \displaybreak[1]\\
    &=\frac{P(c=1,x=x_0)}{\sum_{u' \in U}P(e=1|x=x_0,u=u')P(u=u'|x=x_0)P(x=x_0)} \quad, \nonumber \displaybreak[1]
% \end{split}
\end{align}
% &=\frac{P(c=1,x=x_0)}{\sum_{u' \in U}P(e=1,x=x_0,u=u')} \\
where the ``$do(\cdot)$'' operator means causal intervention, that is, we intervene to make the variable take a certain value rather than its observed value. The second step utilizes the second rule of $do$-calculus \cite{Judea1995docalculus}: because all the backdoor paths from $e$ to $r$ are blocked given the variable $x$, the $do(\cdot)$ operator can be removed. Building on the above causal inference, the local ranking loss can be estimated as follows:
% given a click list $\vec{c}$ from a session of query $q$, we estimate the local ranking loss as follows:
% user-aware estimator
\begin{align}
\begin{split}
    % & \hat{l}_{user-aware}(S, q|N_q) =\frac{1}{|N_q|}\sum_{n \in N_q} \hat{l}_{user-aware}(S, q|\vec{c_n}) \\
    \hat{l}_{user-aware}(S|q,\vec{c}) = \sum_{d\in \pi_q} \frac{\lambda(d|\pi_q)c(d)}{\sum_{u \in U} P(e(d)=1|k(d),u)P(u|q)}
\end{split}
\end{align}

In the following, we theoretically prove that the user-aware estimator is unbiased if each relevant document has a positive probability of being examined. In addition, our user-aware estimator has a lower variance than the straightforward estimator.

\begin{theorem}(Unbiasedness). 
\label{theorem_user-aware_unbiasedness}
    Suppose $\forall d:r(d)=1, \sum_{u \in U}P(e=1|k(d),u)P(u|q) > 0$, then $\mathbb{E}_{u,e} \left[ \hat{l}_{user-aware}(S| q,\vec{c}) \right] = l_{ideal}(S|q)$
\end{theorem}
\begin{proof}
% By taking expectation on $\hat{l}_{user-aware}(S|q,\vec{c})$, we have:
\begin{align}
% \begin{split}
    \mathbb{E}&_{u,e} \left[ \hat{l}_{user-aware}(S|q,\vec{c}) \right] \nonumber \displaybreak[1]\\
    &= \mathbb{E}_{u,e} \left[\sum_{d\in \pi_q} \frac{\lambda(d|\pi_q)c(d)}{\sum_{u' \in U} P(e(d)=1|k(d),u')P(u'|q)} \right] \nonumber \displaybreak[1]\\
    &= \sum_{d\in \pi_q} \mathbb{E}_{u,e} \left[ \frac{e(d) \lambda(d|\pi_q) r(d)}{\sum_{u' \in U} P(e(d)=1|k(d),u')P(u'|q)} \right] \nonumber \displaybreak[1]\\
    &= \sum_{d\in \pi_q}  \mathbb{E}_{e} \left[ \frac{\sum_{u \in U} e(d) \lambda(d|\pi_q) r(d) P(u|q)}{\sum_{u' \in U} P(e(d)=1|k(d),u')P(u'|q)} \right] \nonumber \displaybreak[1]\\
    &= \sum_{d\in \pi_q}  \frac{\sum_{u \in U}  P(e(d)=1|k(d),u) P(u|q)\lambda(d|\pi_q) r(d) }{\sum_{u' \in U} P(e(d)=1|k(d),u')P(u'|q)} \nonumber \displaybreak[1]\\
    &= \sum_{d\in \pi_q}  \lambda(d|\pi_q) r(d) = l_{ideal}(S, q) \nonumber \displaybreak[1]
% \end{split}
\end{align}
The second step utilizes the examination hypothesis in Eq.~(\ref{eq:EH}).
\end{proof}

\begin{lemma} 
\label{lemma}
Suppose there is a finite set $ T = \{t_1, t_2, ..., t_m\}$, where $\forall i \in [1,m], 1 \geq t_i > 0$. We have $\sum_{i=1}^{m} \frac{1-t_i}{t_i} \geq \sum_{i=1}^{m} \frac{t_i(1-t_i)}{(\frac{1}{m}\sum_{i=1}^{m}t_i)^{2}}$. 
\end{lemma}

\begin{proof}
\begin{align}
% \begin{split}
    \sum_{i=1}^{m} \frac{1-t_i}{t_i} &= \sum_{i=1}^{m} \frac{1}{t_i} - m \geq \frac{m^2}{\sum_{i=1}^{m} t_i} - m \nonumber \displaybreak[1] \\
    & = \frac{m^2\sum_{i=1}^{m} t_i - m(\sum_{i=1}^{m} t_i)^2}{(\sum_{i=1}^{m} t_i)^2} \geq \frac{m^2\sum_{i=1}^{m} t_i - m^2\sum_{i=1}^{m} t_i^2} {(\sum_{i=1}^{m} t_i)^2} \nonumber \displaybreak[1] \\
    & =  \frac{\sum_{i=1}^{m} t_i - \sum_{i=1}^{m} t_i^2} {(\frac{1}{m}\sum_{i=1}^{m}t_i)^{2}} = \sum_{i=1}^{m} \frac{t_i(1-t_i)}{(\frac{1}{m}\sum_{i=1}^{m}t_i)^{2}} \nonumber \displaybreak[1]
% \end{split}
\end{align}
The second and fourth steps leverage the mean value inequality.
\end{proof}

\begin{theorem}(Lower variance). 
\label{theorem}
Suppose that $\forall d\in \pi_q, \\ \sum_{u \in U} P(e(d)=1|k(d),u)P(u|q) \geq  \frac{1}{|N_q|}\sum_{(\vec
c, u) \in N_q} P(e(d)=1|k(d),u)$, then $\mathbb{V}[\hat{l}_{user-aware}(S|q,N_q)] \leq \mathbb{V}[\hat{l}_{straight}(S|q,N_q)]$. \footnote{Note that this condition is always satisfied in our implementation since we count the user frequencies in the click logs to estimate the user distributions $P(u|q)$.}
\end{theorem}

\begin{proof}
We first derive the variance of our user-aware estimator:
\begin{align}
% \begin{split}
    \mathbb{V}& \left[ \hat{l}_{user-aware}(S|q,N_q) \right] = \mathbb{V} \left[ \frac{1}{|N_q|}\sum_{(\vec{c},u) \in N_q} \hat{l}_{user-aware}(S|q,\vec{c}) \right] \nonumber \displaybreak[1] \\
    &= \frac{1}{|N_q|^2} \sum_{(\vec{c},u) \in N_q}  \mathbb{V} \left[ \sum_{d\in \pi_q} \frac{\lambda(d|\pi_q)c(d)}{\sum_{u' \in U} P(e(d)=1|k(d),u')P(u'|q)} \right] \nonumber \displaybreak[1] \\
    &= \frac{1}{|N_q|^2} \sum_{d\in \pi_q}\sum_{(\vec{c},u) \in N_q} \frac{\lambda^2(d|\pi_q)}{(\sum_{u' \in U} P(e(d)=1|k(d),u')P(u'|q))^{2}}  \mathbb{V} \left[  c(d) \right] \nonumber \displaybreak[1]\\
    &= \frac{1}{|N_q|^2} \sum_{d\in \pi_q: r(d)=1}\sum_{(\vec{c},u) \in N_q} \nonumber \displaybreak[1]\\
    &\frac{\lambda^2(d|\pi_q) P(e(d)=1|k(d),u) (1 - P(e(d)=1|k(d),u))}{(\sum_{u' \in U} P(e(d)=1|k(d,u')P(u'|q))^{2}} \quad, \nonumber \displaybreak[1]
% \end{split}
\end{align}
where the last step utilizes the examination hypothesis in Eq.~(\ref{eq:EH}). For simplicity, let $p(d,u)$ denote $P(e(d)=1|k(d),u)$ and let $p(d)$ denote $\frac{1}{|N_q|}\sum_{(\vec{c}, u) \in N_q} P(e(d)=1|k(d),u)$. We have:
\begin{align}
% \begin{split}
    \mathbb{V}& \left[ \hat{l}_{user-aware}(S|q,N_q) \right] \nonumber \displaybreak[1]\\
    &= \frac{1}{|N_q|^2} \sum_{d\in \pi_q: r(d)=1}\sum_{(\vec{c}, u) \in N_q} \frac{\lambda^2(d|\pi_q) p(d,u) (1-p(d,u))}{(\sum_{u' \in U} P(e(d)=1|k(d),u')P(u'|q))^{2}} \nonumber \displaybreak[1]\\
    &\leq \frac{1}{|N_q|^2} \sum_{d\in \pi_q: r(d)=1}\sum_{(\vec{c}, u) \in N_q} \frac{\lambda^2(d|\pi_q) p(d,u) (1-p(d,u))}{p^2(d)} \nonumber \displaybreak[1]\\
    &\leq \frac{1}{|N_q|^2} \sum_{d\in \pi_q: r(d)=1}\sum_{(\vec{c}, u) \in N_q} \frac{\lambda^2(d|\pi_q) (1-p(d,n))}{p(d,n)} \nonumber \displaybreak[1]\\
    &= \mathbb{V} \left[ \hat{l}_{straight}(S,|q,N_q) \right] \nonumber \displaybreak[1]
% \end{split}
\end{align}
The second step uses the condition and the third step leverages Lemma~\ref{lemma}.
\end{proof}

\vspace{-4mm}

\textbf{Remarks (1)}. From a computational perspective, the essential difference between our user-aware estimator and IPS-PBM lies in the calculation of the propensity for each query-document pair $P(e=1|q,d)$: IPS-PBM leverages all users' average examination probability, i.e., $\sum_{u \in U} P(e=1|k(d),u)P(u)$; while we average the examination probabilities of the users who issued the query, i.e., $\sum_{u \in U} P(e=1|k(d),u)P(u|q)$. In the case of user-aware ULTR, since $u$ and $q$ are not independent, IPS-PBM is biased. Besides, if users only have personalized browsing behaviors but no personalized querying behaviors, our user-aware estimator degrades to IPS-PBM.

\textbf{Remarks (2).} Figure~\ref{fig:causal_graphs_ipsupbm}b and Figure~\ref{fig:causal_graphs_ipsupbm}c illustrate the causal interventions simulated by the straightforward and user-aware estimators, respectively. These two methods block the backdoor paths from $e$ to $c$ in different ways. By identifying the specific user corresponding to each session and utilizing his/her examination probabilities for different positions as propensities, the straightforward estimator cuts off the paths $u \rightarrow e$ and $k \rightarrow e$. On the contrary, by estimating the user distribution under each query and aggregating these users' examination probabilities for different positions, our user-aware estimator cuts off the paths $u \rightarrow q$ and $k \rightarrow e$. In summary, the user-aware estimator balances the examination distributions over documents under each query rather than under each session and thus reduces the variance.

\textbf{Remarks (3).} The user-aware estimator needs to estimate two sets of parameters: user distributions under each query and users' personalized examination probabilities. The user distributions can be estimated by simply counting the users' frequencies for each query. Besides, it is feasible to estimate each user's examination probabilities by simply extending previous examination estimation methods (i.e., online randomization experiments \cite{joachims2017ips} or offline approaches like Regression-EM\cite{wang2018regression-EM} to utilize user-dependent examination parameters. 

\section{Simulation-based Experiments}

In this section, we conduct simulation-based experiments to verify the effectiveness of our user-aware estimator. In general, we aim to answer the following research questions (RQ):
\begin{itemize}
    \item {RQ1:} How effective is the user-aware estimator in mitigating personalized bias?
    \item {RQ2:} How does the user-aware estimator perform when the number of simulated training sessions varies?
    % \item {RQ3:} How does the user-aware estimator perform when the number of user clusters varies?
    \item {RQ3:} How does the user-aware estimator perform when the differences between users' behaviors vary?
    \item {RQ4:} How effective is the user-aware estimator when the examination propensities are also estimated based on user behavior logs?
\end{itemize}
% We describe the detailed experimental setup in Appendix~\ref{experimental_setup}.

\subsection{Experimental Setup}
\label{sec:exp_setup}
\textbf{Datasets}. We use two public datasets Yahoo! LETOR set~1 \cite{Yahoo!} and Baidu-ULTR \cite{zou2022Baidu-ULTR} to simulate user behavior logs. Both of these datasets are collected from commercial search engines. Specifically,
Yahoo! LETOR set~1 is a learning-to-rank benchmark dataset and contains 29,336 queries with 710k documents. Each query-document pair has a 5-level human relevance annotation (0-4) and 700-dimensional ranking features. We follow the same data split of training, validation, and testing in the Yahoo! LETOR set 1. 
The Baidu-ULTR dataset is a newly proposed dataset for ULTR. Although it provides real click logs, it does not contain any user information. Thus, we still need to use its 10,402 annotated queries with 757k documents to simulate clicks with user IDs. Each query-document pair has a 5-level human relevance annotation and 782-dimensional ranking features. We randomly divide these annotated queries into training, validation, and test sets in a ratio of 7:1:2.

\textbf{Production ranker}. 
Following Joachims et al. \cite{joachims2017ips}, we use 1\% annotation data to train an SVM-rank \cite{joachims2002SVM-rank, joachims2006SVM-rank_linear_time} to generate the initial ranked lists for each training query. Following Ai et al. \cite{ai2018dla}, we truncate the length of each initial ranked list to 10.

\begin{table*}[htbp]
  \caption{Comparison of different estimators with an MLP ranking model learning from 1,000,000 training sessions simulated on Yahoo! LETOR set~1 and Baidu-ULTR. The best-performing estimator is in bold. $+/-$ indicates a result is significantly better or worse (t-test with p-value < 0.05) than the user-aware estimator.}
  \label{tab:unbiasedness_mlp}
  \vspace{-4mm}
  \renewcommand{\arraystretch}{0.88}
  % \resizebox{\textwidth}{!}{
\begin{tabular}{cccccccccc} 
\toprule
Correction Method & nDCG@1 & nDCG@3 &nDCG@5  & nDCG@10 & ERR@1 & ERR@3 & ERR@5 & ERR@10 & MSE\\ \midrule
\multicolumn{10}{c}{Yahoo! LETOR set 1} \\
\hline
production ranker & $0.5681^-$ & $0.6064^-$ & $0.6426^-$ & $0.7031^-$ & $0.2770^-$ & $0.3771^-$ & $0.4031^-$ & $0.4199^-$ & /
\\
ideal & $0.6804^+$ &$0.6905^+$ &$0.7125^+$ &$0.7606^+$ &$0.3448$ &$0.4246^+$ &$0.4466^+$ &$0.4619^+$ & /\\
\hline
naive & $0.6335^-$ &$0.6468^-$ &$0.6714^-$ &$0.7261^-$ &$0.3161^-$ &$0.4012^-$ &$0.4237^-$ &$0.4400^-$ & $0.1153^-$\\
IPS-PBM\cite{joachims2017ips} & $0.6385^-$ &$0.6511^-$ &$0.6774^-$ &$0.7332^-$ &$0.3215^-$ &$0.4026^-$ &$0.4254^-$ &$0.4420^-$ & $0.2212^-$\\
PRS\cite{wang2021prs}& $0.6261^-$ &$0.6400^-$ &$0.6665^-$ &$0.7219^-$ &$0.3042^-$ &$0.3920^-$ &$0.4156^-$ &$0.4321^-$ & / \\
straightforward & $0.6560^-$ &$0.6647^-$ &$0.6903^-$ &$0.7429^-$ &$0.3318^-$ &$0.4115^-$ &$0.4344^-$ &$0.4504^-$ & $0.2226^-$\\
user-aware & \textbf{0.6718} &\textbf{0.6810} &\textbf{0.7056} &\textbf{0.7543} &\textbf{0.3420} &\textbf{0.4217} &\textbf{0.4440} &\textbf{0.4594} & \textbf{0.0593}\\
\hline
\multicolumn{10}{c}{Baidu-ULTR} \\
\hline
production ranker & $0.4540^-$ & $0.4750^-$ & $0.4894^-$ & $0.5222^-$ & $0.1782^-$ & $0.2756^-$ & $0.3060^-$ & $0.3286^-$ & /
\\
ideal & $0.5138$ &$0.5300$ &$0.5437$ &$0.5758$ &$0.1887$ &$0.2882$ &$0.3187$ &$0.3408$ & /\\
\hline
naive& $0.4798^-$ &$0.5064^-$ &$0.5215^-$ &$0.5541^-$ &$0.1760^-$ &$0.2763^-$ &$0.3067^-$ &$0.3291^-$ &$0.1129^-$ \\
IPS-PBM\cite{joachims2017ips}& $0.4970^-$ &$0.5170^-$ &$0.5323^-$ &$0.5640^-$ &$0.1827^-$ &$0.2817^-$ &$0.3124^-$ &$0.3345^-$ &$0.1716^-$\\
PRS\cite{wang2021prs}& $0.4792^-$ &$0.5040^-$ &$0.5158^-$ &$0.5472^-$ &$0.1795^-$ &$0.2770^-$ &$0.3065^-$ &$0.3286^-$ & / \\
straightforward& $0.5127$ &$0.5288$ &$0.5405^-$ &$0.5715$ &$0.1883$ &$0.2874$ &$0.3173$ &$0.3396$ & $0.1014$\\
user-aware& \textbf{0.5183} &\textbf{0.5318} &\textbf{0.5464} &\textbf{0.5744} &\textbf{0.1910} &\textbf{0.2892} &\textbf{0.3200} &\textbf{0.3418} &\textbf{0.0289}\\
\bottomrule
\end{tabular}
% }
\vspace{-5mm}
\end{table*}

\textbf{Click simulation}. Because the public datasets lack real-world user information or user-related features, it is difficult for us to simulate thousands of users with different behaviors. However, previous work on users' behaviors \cite{montanez2014cross-device_search, sharifpour2023analysis-of-query-logs} suggests that users' search behaviors can usually be clustered into several distinct groups based on certain attributes such as education level and device (such as phones and PCs). Besides, Zhang et al. \cite{zhang2022globalorlocalCM} showed that users' browsing behaviors can also be roughly clustered into ten groups. Therefore, we aim to simulate clicks generated by different clusters of users $U = \{u_1, u_2, ..., u_{|U|}\}$ in our simulation-based experiments. The simulation process includes three parts: personalized browsing, personalized query issuing, and shared relevance judgment.

We first model users' personalized browsing behaviors. Following Joachims et al. \cite{joachims2017ips}, we adopt the position-based click model for each $u_i$ to derive the examination probabilities:
\begin{equation}
\label{eq_PBM_exam}
    P(e(d) = 1 | R_q, u_i) =  (\frac{1}{k(d|R_q)})^{\eta_{u_i}} ,
\end{equation}
where $\eta_{u_{i}}$ models the severity of position bias and varies among different user clusters. If not specified, we set $|U|$ to 10 and let $\eta_{u_{i}}$ take the value in $\{2.5, 2.0, 1.8, 1.5, 1.2, 1.0, 0.8, 0.5, 0.2, 0\}$ in sequence by referring to the clustering results of Zhang et al. \cite{zhang2022globalorlocalCM}.

Then, we simulate users' personalized query behaviors. For each user cluster $u_i$, we generate a random multinomial distribution $Pr(u_i) = \{pr_1, pr_2, ..., pr_{|Q|}\}$ over the training queries, where $pr_j$ is a random number, $|Q|$ is the number of unique training queries, and $\sum_j pr_j = 1$.\footnote{Here we let $pr_j$ has a 50\% probability of being 0 to make $Pr(u_i)$ sparser and thus amplify the disparities in user distribution between different queries.} Since Zou et al. \cite{zou2022Baidu-ULTR} showed that empirically there are many more sessions with severe position bias than those with slight position bias, we expect that in our setting more sessions would come from the user clusters with more serious position bias tendencies (larger $\eta_{u_{i}}$). Therefore, we make the number of queries $m$ submitted by users satisfy the following relationship:

\begin{equation}
\forall i < |U|, m_{u_i} = 1.25 \times m_{u_{i+1}} \quad ,
\end{equation}
where $m_{u_i}$ is the number of queries submitted by $u_i$. Next, we can follow the distribution $Pr(u_i)$ to sample $u_i$'s issued queries $Q_{u_i} = \{q_1, q_2, ..., q_{m_{u_i}}\}$.

As for the sampling of shared relevance judgments, similar to Gupta et al. \cite{gupta2023safe}, we apply the following transformation from relevance annotations to relevance probabilities:
% Relevance sampling 
\begin{equation}
\label{eq:relevance_sampling}
P(r(q,d) = 1) = \epsilon + (1 - \epsilon) \times 0.25 \times y(q,d) \quad,
\end{equation}
where $y$ denotes the relevance annotation of query-document pair $(q, d)$, and the $\epsilon$ parameter models click noises so that irrelevant documents (y = 0) have a non-zero probability of being perceived as relevant. Following Ai et al. \cite{ai2018dla}, we fix $\epsilon$ to 0.1 in our experiments.

Finally, users' clicks are simulated following the examination hypothesis (EH) shown in Eq.~(\ref{eq:EH}). Please note that the click-through rate at each position of our synthetic clicks on Baidu-ULTR (0.423, 0.158, 0.098, 0.075, 0.063, 0.053, 0.048, 0.044, 0.041 0.039) is very close to the statistical value of its real logs \cite{zou2022Baidu-ULTR}, which can demonstrate the rationality of our click simulation to a certain extent.

\textbf{Baselines}. To verify the effectiveness of our user-aware estimator, we compare it with the following IPS-based ULTR baselines. First, we implement a \emph{naive} baseline that uses the raw click data to train the ranking model in Eq.~(\ref{eq:naive_estimator}). Second, we reproduce the user-oblivious \emph{IPS-PBM} \cite{joachims2017ips} estimator in Eq.~(\ref{eq:IPS_PBM_estimator}). Third, we reproduce the \emph{PRS} \cite{wang2021prs} method that integrates the inverse propensity weight on both the clicked documents and the non-clicked ones to estimate pairwise ranking objectives. Fourth, we implement the \emph{straightforward} estimator in Eq.~(\ref{eq:straightforward estimator}) that is unbiased but suffers from high variance. Besides, we use 100\% annotation data to train an \emph{ideal} ranking model that can be regarded as the skyline. In addition, for RQ4, we further compare our user-aware estimator with the Contextual Position-Based Model (\emph{CPBM}) \cite{fang2019CPBM} that harvests interventions of logging policies to estimate examination probabilities dependent on both position and contextual information.\footnote{For the differences between CPBM and our user-aware method, please see Section~\ref{sec:related work}.}

\textbf{Model training}.
We count the user frequencies in the click logs to estimate the user distributions $P(u|q)$. In answering RQ1-3, following Oosterhuis et al. \cite{oosterhuis2020policy-aware} and Gupta et al. \cite{gupta2023safe}, we assume that the users' examination probabilities are known (i.e., the parameters used in Eq.~(\ref{eq_PBM_exam})) for all the estimators.\footnote{The estimation of users' examination probabilities is not the focus of this paper and has been addressed by previous work \cite{wang2016ips, joachims2017ips, wang2018regression-EM, ai2018dla}.} When answering RQ4, we utilize the standard Regression-EM \cite{wang2018regression-EM} algorithm to estimate ``average'' examination probabilities for IPS-PBM and PRS. But for straightforward and user-aware, we implement a personalized Regression-EM algorithm that sets different examination parameters for each user. As for the ranking model, we implement both a linear ranking model as Joachims et al. \cite{joachims2017ips} and a Multilayer Perceptron (MLP) ranking model as Ai et al. \cite{ai2018dla}. As for the ranking objective, we adopt a differentiable list-wise loss\footnote{This loss function is used for all the baselines except PRS. For PRS, we use the same pairwise lambda loss as Wang et al. \cite{wang2021prs}.}: 
\begin{equation}
\label{eq:ranking loss}
    l(S|q) = -\sum_{d\in \pi_q}   {\rm log} \frac{e^{s(d)}}{\sum_{z\in \pi_q} e^{s(z)}} \times 0.25 \times y(d) ,
\end{equation}
where $s(d)$ denotes the ranking score for document $d$ outputted by the ranking model $S$, and $y$ denotes the relevance annotation.

The user-aware estimator only increases little computational cost compared with IPS-PBM. For training, our estimator differs from IPS-PBM in the following steps: (1) We calculate the user distributions by counting user frequencies in the click log, which traverses the dataset once. (2) When offline estimating users' examination probabilities, the total amount of computation required for our personalized Regression-EM is similar to that of standard Regression-EM. (3) We calculate user-weighted propensities for each query, which is a simple process and only takes a small amount of computation. The additional cost of these three steps is negligible when compared to the training time of the ranking model itself. For inference, only the trained ranking model is used and there is no additional computational cost.

\textbf{Code.} Our code for click simulation and model implementation is available at a Github repository \href{https://github.com/Diligentspring/Personalized-Bias}{Personalized-Bias}.\footnote{\href{https://github.com/Diligentspring/Personalized-Bias}{https://github.com/Diligentspring/Personalized-Bias.}}

\subsection{RQ1: How effective is the user-aware estimator in mitigating personalized bias?}
Following the described simulation settings, we first generated 1 million synthetic training sessions on both Yahoo! LETOR and Baidu-ULTR datasets, respectively. Then, we utilized different estimators to estimate the ranking loss in Eq.~(\ref{eq:ranking loss}) on the training set and trained the MLP ranking model. Finally, we evaluated their ranking performance on the test set with human relevance annotations. The metrics we used include nDCG\cite{jarvelin2002cumulated} and ERR \cite{chapelle2009expected} at positions 1, 3, 5, and 10. 
Besides, we computed the mean-square error (MSE) between the estimated probabilities of being relevant $P(r=1) = P(c=1)/P(e=1)$ by each debiasing method and the oracle values in Eq.~(\ref{eq:relevance_sampling}) for the training query-document pairs.\footnote{PRS directly estimates pairwise loss, so we cannot compute MSE for PRS.} The results are shown in Table~\ref{tab:unbiasedness_mlp}, where the production ranker refers to the 1\% SVM-rank that generates the initial ranked lists for the training queries. We can see that the user-aware estimator outperforms all the baselines in terms of all the metrics on both semi-synthetic datasets, which demonstrates its effectiveness when facing personalized bias. Although the straightforward estimator is also theoretically unbiased, its performance is worse than the user-aware estimator because of the high variance problem. The user-oblivious estimators IPS-PBM and PRS can not address the personalized bias and thus perform significantly worse than the straightforward and user-aware estimators. 

\begin{figure}[tbp]
    \centering
    % \vspace{-4mm}
    % \includegraphics[width=0.90\linewidth, trim=20 0 20 0, clip]{mitigate_bias.pdf}
    % \vspace{-3mm}
    \includegraphics[width=0.90\linewidth, trim=20 0 20 0, clip]{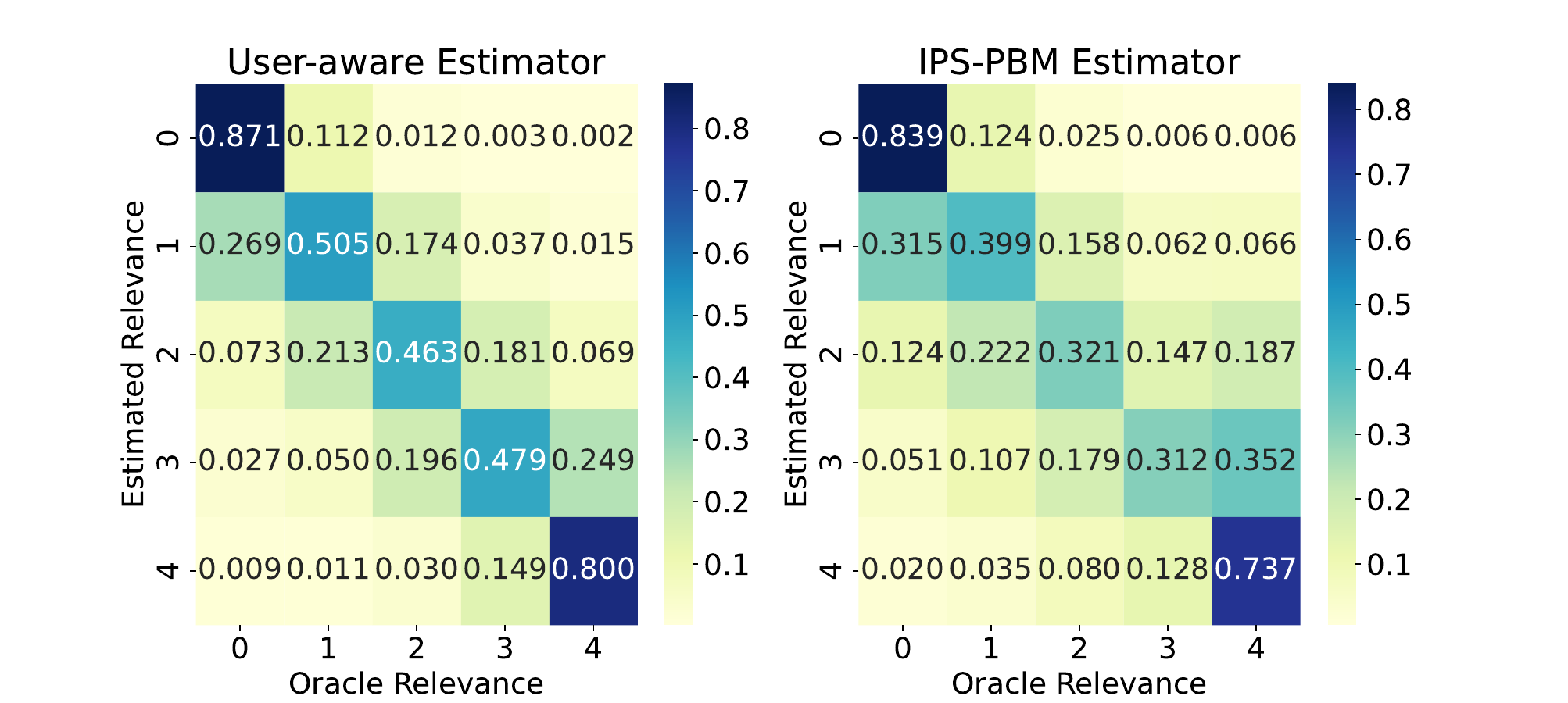}
    \vspace{-4mm}
    \caption{Confusion matrix of estimated relevance for training query-document pairs on Yahoo! LETOR.}
    \label{fig:heatmap}
    \vspace{-5mm}
\end{figure}

\begin{figure}[tbp]
    \centering
    % \includegraphics[width=1.0\linewidth]{Yahoo_variance.pdf}
    % \vspace{-3mm}
    \includegraphics[width=0.90\linewidth, trim=15 0 15 20, clip]{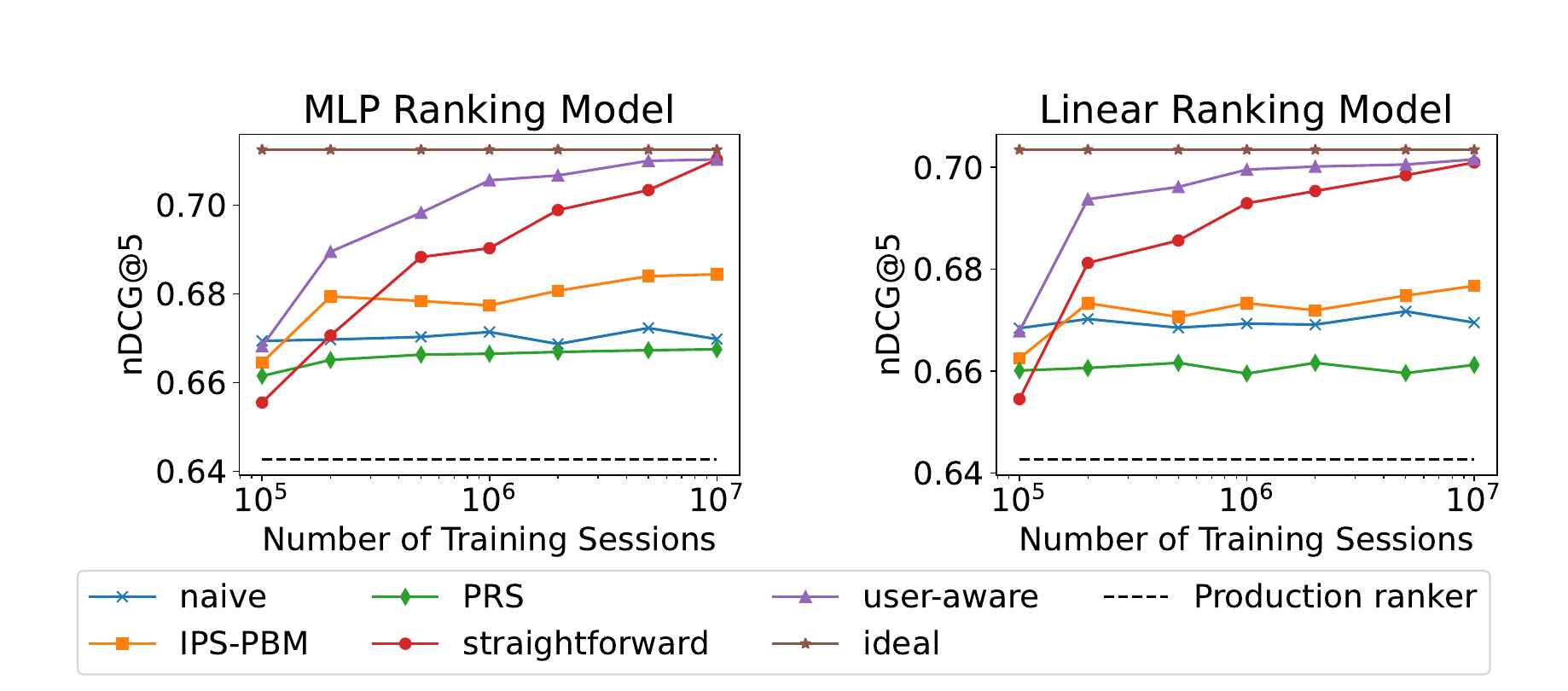}
    \vspace{-4mm}
    \caption{nDCG@5 scores of the estimators learning from various numbers of simulated training sessions on Yahoo!.}
    \label{fig:lower_variance}
    \vspace{-5mm}
\end{figure}

To more intuitively compare the debiasing effect of the user-aware and IPS-PBM estimators, we plot the confusion matrix between their relevance estimates and the oracle values as Figure~\ref{fig:heatmap}.\footnote{The relevance estimates are transformed from the estimated relevant probabilities.} User-aware estimates the relevance more accurately than IPS-PBM, especially for those query-document pairs with relevance annotations between 1 and 3. 

\subsection{RQ2: How does the user-aware estimator perform when the number of simulated training sessions varies?}
We simulated different numbers of training sessions on Yahoo! LETOR and compared the user-aware estimator with the baselines. We trained both the linear and MLP ranking models and evaluated their nDCG@5 scores as shown in Figure~\ref{fig:lower_variance}. When the number of training sessions is larger than 100,000 (which is easy to achieve in Web search), the user-aware estimator always outperforms the user-oblivious and naive baselines. As the number of training sessions increases, the user-aware and straightforward estimators will gradually converge to the ideal skyline, which indicates that they can obtain unbiased estimates of the ranking objective. It is worth noting that the user-aware estimator converges faster than the straightforward estimator due to its lower variance. 

\subsection{RQ3:How does the user-aware estimator perform when the differences between users' behaviors vary?}

To vary the differences between users’ behaviors, we set up different numbers of user clusters ($|U| \in \{5, 10, 20\}$). when $|U| = 5$,  $\eta_{u_{i}}$ in Eq.~(\ref{eq_PBM_exam}) takes value in $\{2.5, 2, 1, 0.8, 0\}$. When $|U| = 20$, $\eta_{u_{i}}$ takes value in $\{2.5, 2.4, 2.2, 2, 1.9, 1.8, 1.6, 1.5, 1.4, 1.2, 1.1, 1, 0.9, 0.8, 0.6, 0.5,$ $ 0.4, 0.2, 0.1, 0\}$. We still simulated 1 million training sessions on Yahoo! LETOR and trained the MLP ranking model. Table \ref{tab:user_cluster} shows the nDCG@5 scores of different estimators. When dealing with different numbers of user clusters, our user-aware estimator always achieves the best performance, which indicates its robustness to different degrees of personalized bias. And please note that our simulation setting of $|U| \in \{5, 10, 20\}$ can cover many realistic scenarios, such as that in \cite{montanez2014cross-device_search} and \cite{sharifpour2023analysis-of-query-logs} where the number of user clusters is no more than 6.

\begin{table}[tbp]
    \centering
    \caption{nDCG@5 scores of different estimators dealing with various user clusters. Same convention as in Table~\ref{tab:unbiasedness_mlp}.}
    \vspace{-3mm}
    \label{tab:user_cluster}
    \renewcommand{\arraystretch}{0.85}
    \begin{tabular}{cccc}
    \toprule
        \multirow{2}{*}{Correction Method} & \multicolumn{3}{c}{User Cluster Number} \\
        \cline{2-4}
        & 5 & 10 & 20\\
        \midrule
        production ranker & \multicolumn{3}{c}{$0.6426^-$} \\
        % \cline{2-3}
        ideal & \multicolumn{3}{c}{$0.7125^+$} \\
        \hline
        naive & $0.6758^-$ & $0.6714^-$ & $0.6520^-$ \\
        IPS-PBM\cite{joachims2017ips} & $0.6835^-$ &$0.6774^- $  & $0.6833^- $\\
        PRS\cite{wang2021prs} & $0.6721^-$ & $0.6665^-$ & $0.6594^-$ \\
        straightforward & $0.6929^-$ & $0.6903^-$ & $0.6925^-$ \\
        user-aware & $\textbf{0.6980}$ & $\textbf{0.7056}$ & $\textbf{0.7042}$ \\
        \bottomrule
    \end{tabular}
    \vspace{-3mm}
\end{table}

\subsection{RQ4: How effective is the user-aware estimator when the examination propensities are also estimated based on user behavior logs?}
\label{Sec:RQ4}
% \subsection{Adaptability to Estimated Examination Probabilities (Answer for RQ3)}

\begin{table}[tbp]
    \centering
    \caption{nDCG@5 scores of different estimators with estimated examination probabilities on Yahoo! LETOR. Same convention as in Table~\ref{tab:unbiasedness_mlp}.}
    % The percentages with downward arrows indicate the decrease rates compared to using the oracle examination probabilities.
    \label{tab:Regression-EM}
    \vspace{-3mm}
    \renewcommand{\arraystretch}{0.85}
    \begin{tabular}{ccc}
    \toprule
        \multirow{2}{*}{Correction Method} & \multicolumn{2}{c}{Ranking Model} \\
        \cline{2-3}
        & MLP & Linear\\
        \midrule
        % production ranker & \multicolumn{2}{c}{$0.6426^-$} \\
        % \cline{2-3}
        ideal & $0.7125^+$ & $0.7034^+$ \\
        \hline
        naive & $0.6714^-$ & $0.6693^-$ \\
        IPS-PBM\cite{joachims2017ips} &$0.6585^- $  & $0.6415^- $\\
        PRS\cite{wang2021prs} & $0.6659^-$ & $0.6588^-$ \\
        CPBM\cite{fang2019CPBM} & $0.6812^-$ & $0.6736^-$ \\
        straightforward & $0.6903^-$ & $0.6899^-$ \\
        user-aware & $\mathbf{0.6998}$ & $\mathbf{0.6927}$ \\
        \bottomrule
    \end{tabular}
    \vspace{-5mm}
\end{table}

Since online randomization experiments may not be possible in practical applications, the personalized examination probabilities may be unknown and need to be estimated from click logs. To evaluate the user-aware estimator in fully offline settings, we still simulated 1 million sessions on Yahoo! LETOR and utilized a personalized Regression-EM to estimate each user cluster's position-dependent examination probabilities. For IPS-PBM and PRS, we utilize the standard Regression-EM algorithm to estimate user-oblivious examination propensities. When reproducing the contextual PBM (CPBM) \cite{fang2019CPBM} baseline, we used one-hot user ID as the context vector and made plenty of interventions on the initial ranked lists to enable its examination estimation. The nDCG@5 scores of different ULTR methods in this setting are shown in Table~\ref{tab:Regression-EM}. Without the oracle examination probabilities, our user-aware estimator still achieves the best performance, which demonstrates its applicability in fully offline settings. In addition, even using additional intervention data, CPBM performs worse than our user-aware estimator because it uses session-level propensities to reweigh clicks and thus suffers from high variance.

% Besides, although the user-PBM baseline also considers users' different browsing behaviors, it performs worse than straightforward and user-aware due to the inconsistency between its click prediction loss and the ranking objective.

\section{Real-world Experiments}
\label{sec:real-world}

In this section, we conducted experiments on a real-world dataset collected from a commercial search engine. The training set contains 299,417 search sessions from 5,168 users, involving 130,788 unique queries and 1,306,982 documents. We further obtain 5-level relevance annotations for 360 queries and 14,762 corresponding query-document pairs. Each query-document pair is represented by 90-dimensional ranking features. We split the annotated queries into validation and test sets in a ratio of 1:6.

\begin{figure}[tbp]
    \centering
    % \vspace{-4mm}
    % \includegraphics[width=0.90\linewidth, trim=20 0 20 0, clip]{mitigate_bias.pdf}
    \vspace{-5mm}
    \begin{subfigure}[b]{0.43\linewidth}
        \centering
        \includegraphics[width=0.85\linewidth, trim=5 5 20 20, clip]{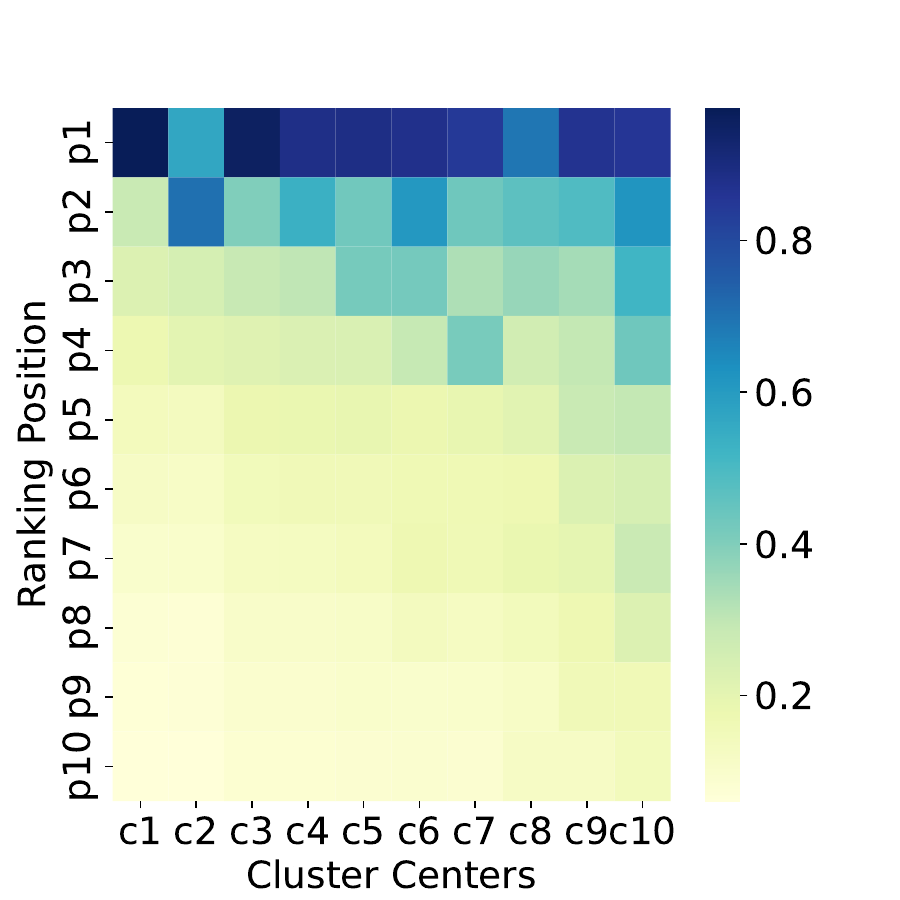}
        \vspace{-2mm}
        \caption{}
        \label{fig:user_cluster}
    \end{subfigure}
    \hfill
    \begin{subfigure}[b]{0.55\linewidth}
        \centering
        \includegraphics[width=\linewidth, trim=5 5 5 5, clip]{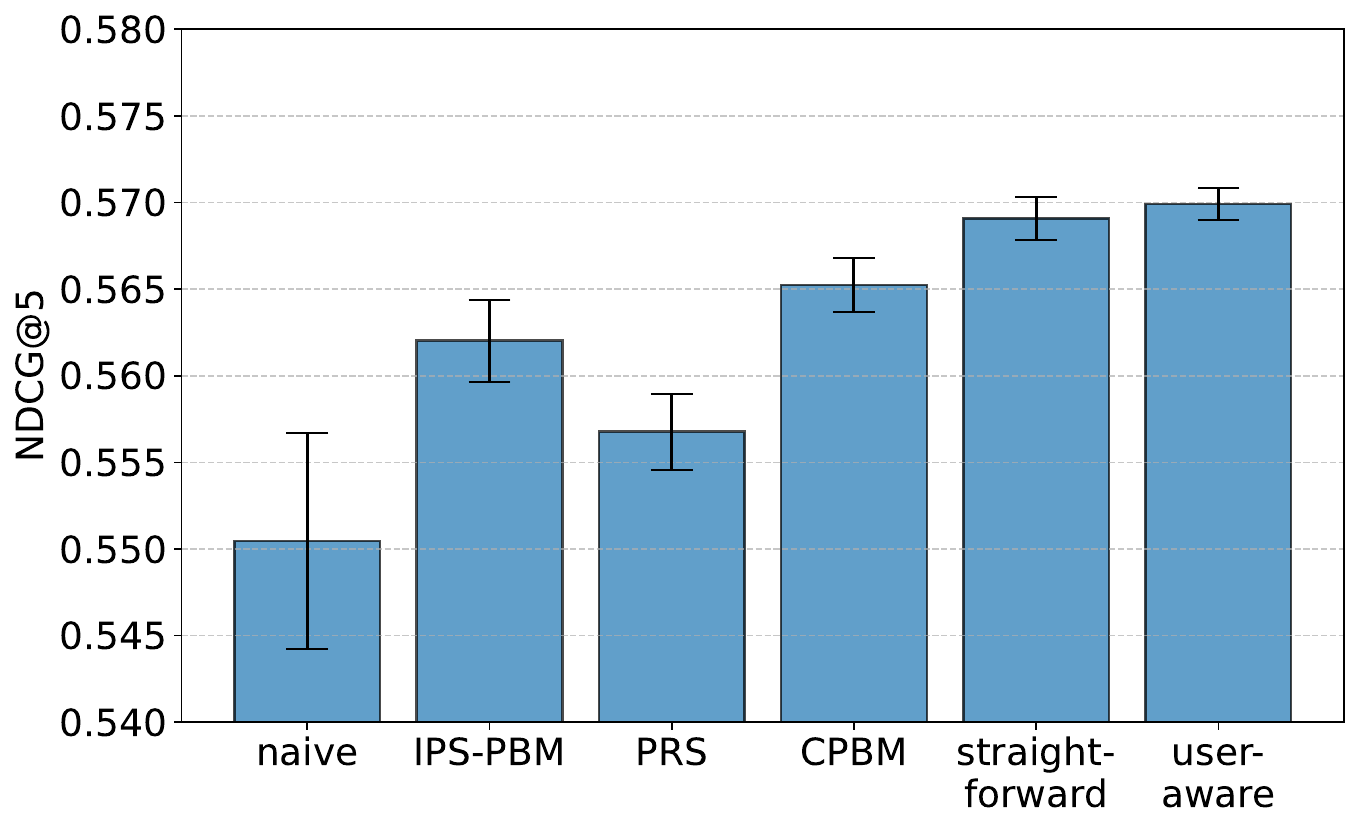}
        \vspace{-5mm}
        \caption{}
        \label{fig:ndcg5_comparison}
    \end{subfigure}
    \vspace{-5mm}
    \caption{(a) Clusters of users' examination propensities. The user clusters are sorted according to the mean position of examination. ‘pi’ denotes the $i$-th ranking position. ‘ci’ denotes the cluster center of the $i$-th user group. (b) nDCG@5 scores with 95\% confidence intervals of different estimators on the real-world dataset. Results are averaged over 5 runs.}
    \vspace{-5mm}
\end{figure}

We first utilized the same personalized Regression-EM \cite{wang2018regression-EM} algorithm in Section~\ref{Sec:RQ4} to estimate the position-dependent examination propensities of each user. To visualize the differences in users' browsing behaviors, we used the K-Means \cite{lloyd1982Kmeans} algorithm to cluster users' examination propensities in Figure~\ref{fig:user_cluster} and show the average propensities for each cluster. We can observe that some user clusters like \verb|c1| are severely affected by the position and mainly focus on the first result. Some user clusters like \verb|c10| tend to examine more results and are more likely to examine the lower-ranked documents than c1. In addition, some user clusters, such as \verb|c2| and \verb|c7|, may prefer to examine documents at middle positions like \verb|p2| and \verb|p4|. These results are consistent with those of Zhang et al. \cite{zhang2022globalorlocalCM} and demonstrate that users have personalized browsing behaviors. 

Then, we estimated the ranking loss in Eq.~(\ref{eq:ranking loss}) on the training set with different ULTR methods and trained the MLP ranking model. The nDCG@5 scores on the test set are shown in Figure~\ref{fig:ndcg5_comparison}. Our proposed user-aware estimator still outperforms all the baselines, which verifies its effectiveness in mitigating personalized bias in real-world scenarios. 

\section{Related Work}
\label{sec:related work}
\textbf{Personalized Search and Browsing Behaviors.} Many studies show that users have personalized search and browsing behaviors. Montanez et al. \cite{montanez2014cross-device_search} showed that users using different Internet devices (including PC, smartphone, tablet, and gaming console) will have different topic interests and issue different queries. Sharifpour et al. \cite{sharifpour2023analysis-of-query-logs} applied cluster analysis on users' behavior logs and divided six user groups with different search behaviors based on their domain expertise and intent. Zhang et al. \cite{zhang2022globalorlocalCM} argued that different users probably examine the documents in personalized ways because personalizing the users' examination parameters improves the performance of click models in predicting real-world clicks.

\textbf{Mitigating Biases in User Behavior Data.} Vardasbi et al. \cite{vardasbi2020cm-ips} extended the IPS-PBM estimator to mitigate position bias under the assumption of cascade-based models such as dependent click model (DCM) \cite{guo2009dcm} and click chain model (CCM) \cite{guo2009CCM}. Agarwal et al. \cite{agarwal2019trustbias} and Vardasbi et al.\cite{vardasbi2020affine} studied the trust bias that users may overestimate the relevance of the top-ranked results due to their trust in search engines' ranking performance. Besides, Wang et al. \cite{chen2021iobm} proposed an embedding-based method to address the interactional observation bias, which considers the interactions between examinations/clicks when estimating the examination probabilities.

The most similar work to ours is that of Fang et al. \cite{fang2019CPBM}, which focused on the context-dependent examination estimation problem and argued that examination probabilities should also depend on contextual information related to each query (such as query length, candidate set size, and user age). We discuss the differences between \cite{fang2019CPBM} and our work as follows: (1) We are the first to analyze the causal effect of the user variable in ULTR and identify the personalized bias caused by the combination of users’ personalized search and examination behaviors. In contrast, Fang et al. \cite{fang2019CPBM} mainly focused on the impact of query-related features on examination probabilities, and their real-world experiments did not use any user features. As for their simulation experiments, the user-related features of each session were randomly sampled from a normal distribution, which would not form personalized bias because the user variable $u$ does not affect query $q$. (2) When estimating ranking objectives, Fang et al. \cite{fang2019CPBM} balances the examination distributions at the session level like IPS-PBM \cite{joachims2017ips}, so it is similar to the straightforward estimator and suffers from high variance. To overcome this problem, in this paper, we propose a novel user-aware estimator utilizing $do$-calculus and theoretically prove its unbiasedness and lower variance. (3) Fang et al. \cite{fang2019CPBM} estimated examination bias by harvesting the intervention done by multiple logging policies, which would be affected if there is only one deterministic logging policy. In contrast, we used Regression-EM to offline estimate examination probabilities, which is more general and easier to extend to other user behavior models.

% and in turn, illustrates the significance of our causal analysis

\textbf{Propensity Estimation in Unbiased Learning to Rank.} Joachims et al. \cite{joachims2017ips} estimated the relative examination probability of two positions $P(e=1|k_1) /P(e=1|k_2)$ by conducting an online experiment that randomly swaps the results at rank $k_1$ with the results at rank $k_2$ and counts the relative click-through rates $P(c=1|k_1) /P(c=1|k_2)$. However, the online experiments can negatively affect users’ search experience. In response to this shortcoming, Wang et al. \cite{wang2018regression-EM} proposed the regression-based Expectation-Maximization (Regression-EM) algorithm to estimate examination probabilities from historical click logs. Regression-EM improves the standard EM algorithm by using a regression function to estimate the relevance of query-document pairs with their features.
\vspace{-2mm}

\section{Conclusion and Future Work}
In this paper, we formulate the user-aware ULTR problem from a causal view and identify an unsolved personalized bias caused by users' personalized search and browsing behaviors. To mitigate the personalized bias, we propose a novel user-aware inverse-propensity-score estimator that balances the examination distributions over documents under each query rather than under each session. Theoretical analyses show that our user-aware estimator is unbiased under certain mild assumptions and has a lower variance than the straightforward extension of IPS-PBM. Extensive experiments on semi-synthetic and real-world datasets demonstrate the superiority of our user-aware estimator over the existing ULTR baselines and the straightforward estimator in dealing with personalized bias. Our work is a first step toward considering users' personalized behaviors in ULTR. In the future, it would be interesting to solve the user-aware ULTR problem in scenarios where users have more complex click behaviors. For example, we can further assume that users also have personalized relevance judgments or they may browse the SERPs in a more complicated manner.
% Unbiased learning to rank has tremendous practical application value and has received lots of attention in the IR community.

\begin{acks}
This research was supported by the Natural Science Foundation of China (61902209, 62377044), Intelligent Social Governance Platform, Major Innovation \& Planning Interdisciplinary Platform for the ``Double-First Class" Initiative, Renmin University of China, the Fundamental Research Funds for the Central Universities, the Research Funds of Renmin University of China (22XNKJ15), and Beijing Nova Program. 

% Any opinions, findings, conclusions, or recommendations expressed in this material are those of the authors and do not necessarily reflect those of the sponsors.
\end{acks}

\section{GenAI Usage Disclosure}
In this work, GenAI was only utilized to help to annotate the relevance labels on the test set of the real dataset in Section~\ref{sec:real-world}. We manually refined the annotation process of GenAI carefully, so that the Cohen's Kappa coefficient between the GenAI and manual annotators reached 0.6443, which is quite reliable.

\bibliographystyle{ACM-Reference-Format}
\balance
\bibliography{sample-base}

\end{document}